\documentclass[journal]{IEEEtran}

\usepackage{latexsym,amsmath,amssymb}
\def\F {{\mathbb{F}}}

\def\cX{{\mathcal X}}

\def\cL{{\mathcal L}}

\def\ba{{\bf a}}
\def\bb{{\bf b}}
\def\bc{{\bf c}}

\def\bv{{\bf v}}

\def\Ga{{\alpha}}

\def\res{{\rm res}}
\def\beq{\begin{equation}}
\def\eeq{\end{equation}}

\def\Go{{\omega}}
\def\GO{{\Omega}}
\def\res{{\rm res}}
\def\Pin {{P_{\infty}}}

\newtheorem{thm}{Theorem}[section]

\newtheorem{prop}[thm]{Proposition}
\newtheorem{lem}[thm]{Lemma}
\newtheorem{cor}[thm]{Corollary}
\numberwithin{equation}{section} \newtheorem{rem}[thm]{Remark}
\newtheorem{exm}[thm]{Example}

\begin{document}

\title{ Quantum Stabilizer Codes from  Maximal Curves}

\author{Lingfei~Jin
\thanks{L. F. Jin  is with Division of
Mathematical Sciences, School of Physical and Mathematical Sciences,
Nanyang Technological University, Singapore 637371, Republic of
Singapore (email: lfjin@ntu.edu.sg). This work is supported in part by the Singapore A*STAR SERC under Research Grant 1121720011.}}

%\thanks{The work is partially supported by ...}}

\maketitle

\begin{abstract} A curve attaining the Hasse-Weil bound is called a maximal curve. Usually classical error-correcting codes obtained from a maximal curve have good parameters. However, the quantum stabilizer  codes obtained from such classical error-correcting codes via Euclidean or Hermitian self-orthogonality do not always possess good parameters. In this paper, the Hermitian self-orthogonality of  algebraic geometry codes obtained from two maximal curves  is investigated. It turns out that the stabilizer quantum codes produced from such Hermitian self-orthogonal classical codes have good parameters.
%%%%%%%%%%%%%%%%%%%%%%
\end{abstract}


\section{Introduction}
A powerful construction of quantum codes is through classical codes with certain self-orthogonality \cite{Ash Kni, Ket Kla}. Among these self-orthogonalities, the Hermitian orthogonality produces $q$-ary quantum codes from $q^2$-ary classical error-correcting codes, therefore Hermitian self-orthogonal classical codes may give rise to good  quantum stabilizer codes. However, it is more challenging to construct Hermitian self-orthogonal classical codes than Euclidean self-orthogonal classical codes.

A good family of Hermitian self-orthogonal classical codes is from algebraic geometry codes \cite{JLLX,JX,Kim}. For instance, in \cite{JLLX}, a family of Hermitian self-orthogonal generalized Reed-Solomon codes is constructed and consequently a family of quantum MDS codes is produced. However, the situation is not always like this. For instance, if we consider the quantum codes produced from the Hermitian self-orthogonal classical codes based on the Hermitian curves, the parameters of these quantum codes are not satisfactory (see \cite{Sar}). To show that an algebraic geometry code is Euclidean or Hermitian self-orthogonal, it is essential to construct a proper differential that satisfies certain condition (See Proposition \ref{2.3}). This is usually challenging, in particular, for  the Hermitian self-orthogonality.

In this paper, we first study two maximal curves and the corresponding classical algebraic geometry codes. A useful result is that we are able to construct a suitable differential to describe their Euclidean dual codes. Then via their Euclidean self-orthogonality, we can show that these codes are Hermitian self-orthogonal for certain parameters. Finally, we apply the stabilizer method \cite{Ash Kni} to obtain  quantum codes which have good parameters or even better parameters compared with those in \cite{Br12, Gr12}.

The paper is organized as follows. In Section 2, we briefly introduce some background on algebraic curves and algebraic geometry codes. Section 3 is devoted to two maximal curves and the corresponding algebraic geometry codes with Hermitian self-orthogonality. In Section 4, we produce good quantum codes from Hermitian self-orthogonal classical codes given in Section 3. Comparisons are given as well to show that quantum codes obtained from our construction are indeed good.

\section{Preliminary}
In this section, we briefly introduce some notations and results on algebraic curves and algebraic geometry codes. The reader may refer to \cite{Fu, St93} for the details.

Let $\cX$ be a smooth, projective, absolutely irreducible curve of genus $g$ defined over $K$, where $K$ is a finite field. We denote by $K(\cX)$ the function field of $\cX$. An element of $K(\cX)$ is called a function. The normalized discrete valuation corresponding to a point $P$ of $\cX$ is written as $\nu_P$. For every nonzero element $f$ of $K(\cX)$, we can define a principal divisor ${\rm div}(f):=\sum_{P}\nu_P(f)P$.

For a divisor $G$, the Riemann-Roch space associated to $G$ is defined by
\[\cL(G)=\{f\in K(\cX)\setminus\{0\}:{\rm div}(f)+G\geq0\}\cup\{0\}.\]
Then $\cL(G)$ is a finite-dimensional vector space over $K$ and we denote its dimension by $\ell(G)$.

Let $\Omega$ denote the differential space of $\cX$. For any nonzero differential $\omega$, we can associate a canonical divisor ${\rm div}(\omega):=\sum_{P}\nu_P(\omega)P$. All canonical divisors are equivalent and have degree $2g-2$. For a divisor $G$, we define
\[\Omega(G)=\{\omega\in \Omega\setminus\{0\}:\;{\rm div}(\omega)\ge G\}\]
and denote the dimension of $\Omega(G)$ by $i(G)$. Then one has
\[i(G)=\ell(H-G),\]
where $H$ is a canonical divisor.

The Riemann-Roch Theorem says that
\[\ell(G)=\deg(G)-g+1+\ell(H-G),\]
where $H$ is any canonical divisor.

Before introducing algebraic geometry codes, let us fix some basic notations.
Let $P_1,\dots,P_n$ be pairwise distinct $K$-rational points of $\cX$ and $D=P_1+\dots+P_n$. Choose a divisor $G$ on $\cX$ such that ${\rm supp}( G)\cap {\rm supp}( D)=\varnothing$. Then $\nu_{P_i}(f)\geq0$ for all $1\leq i\leq n$ and any $f\in \cL(G)$.

Consider the following two maps
\[\Psi:\cL(G)\rightarrow K^n,\quad f\mapsto(f(P_1),\dots,f(P_n))\]
and
\[\Phi:\Omega(G-D)\rightarrow K^n,\quad \omega\mapsto(\res_{P_1}(\omega),\dots,\res_{P_n}(\Go)),\]
where $\res_{P_i}(\Go)$ denotes the residue of $\Go$ at $P_i$ (see \cite[Chapter 2]{St93}).
 The images of $\Psi$ and $\Phi$ are denoted by $C_\cL(D,G)$ and $C_\GO(D,G)$, respectively.   It is clear that both $C_\cL(D,G)$ and $C_\GO(D,G)$ are linear codes over $K$. They are called  algebraic-geometry codes (or AG codes for short). A nice property is that the Euclidean dual $C_\cL(D,G)^{\perp}$ ($\perp$ denotes the Euclidean dual) of $C_\cL(D,G)$ is $C_\GO(D,G)$ (see \cite[Theorem II.2.8]{St93}).

 Furthermore, we have the following results.

\begin{prop}\label{2.1}(\cite[Theorem II.2.2 and Corollary II.2.3]{St93}) $C_\cL(D,G)$ is an $[n,k,d]$-linear code over $K$ with parameters
\[k=\ell(G)-\ell(G-D),\quad  d\geq n-\deg(G).\]
\begin{itemize}
\item[{\rm (a)}] If $G$ satisfies $\deg(G)<n$, then
$$k=\ell(G)\geq \deg(G)-g+1.$$
\item[{\rm (b)}] If additionally $2g-2<deg(G)<n$, then $k=deg(G)-g+1$.
\end{itemize}
\end{prop}

\begin{prop}\label{2.2}(\cite[Theorem II.2.7]{St93}) $C_\GO(D,G)$ is an $[n,k^{\perp},d^{\perp}]$-linear code over $K$ with parameters
\[k^{\perp}=i(G-D)-i(G),\quad  d^{\perp}\geq \deg(G)-(2g-2).\]
\begin{itemize}
\item[{\rm (a)}] If $G$ satisfies $\deg(G)> 2g-2$, then
$$k^{\perp}=i(G-D)\geq n+g-1-\deg(G).$$
\item[{\rm (b)}] If additionally $2g-2<deg(G)<n$, then $$k^{\perp}=i(G-D)= n+g-1-\deg(G)$$.
\end{itemize}
\end{prop}

To study Euclidean self-orthogonality, we have to investigate the relationship between $C_\cL(D,G)$ and $C_\GO(D,G)$.

\begin{prop}\label{2.3}(\cite[Theorem II.2.10]{St93}) Let $\eta$ be a differential such that $\nu_{P_i}=-1$ and $\res_{P_i}(\eta)=1$ for all $i=1,\dots,n$. Then
\[C_\cL(D,G)^{\perp}=C_\GO(D,G)=C_\cL(D,D-G+{\rm div}(\eta)),\]
where $C_\cL(D,G)^{\perp}$ stands for the Euclidean dual of $C_\cL(D,G)$.
\end{prop}

To obtain good classical AG codes,
 one is interested in the number of $K$-rational points on an algebraic curve.  We denote by $N_K(\cX)$ the  number of $K$-rational points on an algebraic curve $\cX$  over $K$. A celebrated result on the number of $K$-rational points is the Hasse-Weil bound stating that
\[N_K(\cX)\leq |K|+1+2g\sqrt{|K|}.\]

If the number of rational points of a curve $\cX$ achieves the upper bound, i.e., $N_K(\cX)= |K|+1+2g\sqrt{|K|}$,  then $\cX$ is called a maximal curve. A well-known maximal curve is the Hermitian curve over $\F_{q^2}$ defined by the equation $y^q+y=x^{q+1}$, where $\F_{q^2}$ denotes the finite field of $q^2$ elements. Lots of maximal curves can be produced by coverings of the Hermitian curve \cite{GSX}. In the next section, we consider a maximal curve which is also a covering of the Hermitian curve.

\section{AG codes from maximal curves}
Throughout the rest of this paper, we consider the finite field $K=\F_{q^2}$, where $q$ is a power of $2$.
\subsection{AG codes from the first maximal curve}
Let $F=\F_{q^2}(\cX)$ be the function field of $\cX$ over $\F_{q^2}$, where $\cX$ is defined by the  following equation
\[y^2+y=x^{q+1}.\]
The genus $g$ of $\cX$ is $g=q/2$ and the number of rational points is $2q^2+1$. The set of these $2q^2+1$  rational points consists of a point at infinity $P_\infty$ and the other  $2q^2$ ``finite" rational points.

Let $n=2q^2$ and let $\{P_1,\dots, P_n\}$ be all $n$ ``finite" rational points. Put
$D=P_1+\dots+P_n$.

\begin{lem}\label{3.1} For a positive integer $m$,
the Euclidean dual $C_\cL(D,m\Pin)^{\perp}$ of $C_\cL(D,m\Pin)$ is $C_\cL(D,(n+2g-2-m)\Pin)$ .
\end{lem}
\begin{proof}
Consider the differential $\eta=\frac{dx}{x-x^q}$. Then one can verify that ${\rm div}(\eta)=-D+(n+2g-2)\Pin$ and $\res_{P_i}(\eta)=1$ for all $i=1,\dots,n$. Thus, by Proposition \ref{2.3}, we have
\begin{eqnarray*}
C_\cL(D,m\Pin)^{\perp}&=&C_\GO(D,m\Pin)\\
&=&C_\cL(D,D-m\Pin+{\rm div}(\eta))\\
&=&C_\cL(D,(n+2g-2-m)\Pin).\end{eqnarray*}
This completes the proof.
\end{proof}

\begin{rem}
From  Lemma \ref{3.1}, the dual of the AG code $C_\cL(D,m\Pin)$ can be represented as another AG code by choosing suitable differential. Therefore, self-orthogonality of the AG code can be described in the term of the degree of divisor $G$, i.e., $m$ in our case.  However, this is not always the case for other curves. Actually it is a challenging task to find the proper differential needed.
\end{rem}

For simplicity, let us denote by $C_m$ the AG code $C_\cL(D,m\Pin)$. Then, the above result says that $C_m^{\perp}=C_{n+2g-2-m}$. Hence, Lemma \ref{3.1} gives the following result.
\begin{cor}\label{3.2} $C_m$ is Euclidean self-orthogonal if $m\le n/2+g-1$.
\end{cor}

Recall that the Hermitian inner product for two vectors $\ba=(a_1,\dots,a_n), \bb=(b_1,\dots,b_n)$ in
$\F_{q^2}^n$ is defined by
$\langle\ba,\bb\rangle_H:=\sum_{i=1}^na_ib_i^q$. For a linear code $C$ over $\F_{q^2}$,
the {\it Hermitian dual} of $C$ is defined by
\[C^{\perp_H}:=\{\bv\in\F_q^n:\;\langle\bv,\bc\rangle_H=0\ \forall\ \bc\in C\}.\]
Then $C$ is Hermitian self-orthogonal if $C\subseteq C^{\perp_H}$. by the definition of Hermitian self-orthogonality, one can easily obtain a useful fact, namely $C\subseteq C^{\perp_H}$ if and only if $C^q\subseteq C^{\perp}$.

\begin{thm}\label{3.3}
$C_m$ is Hermitian self-orthogonal if  $m\leq 2q-2$.
\end{thm}
\begin{proof} If $m\le 2q-2$, then we have $mq\le n+2g-2-m$. Thus, one has $C_{mq}\subseteq C_{n+2g-2-m}$. Hence, the desired result follows from the fact that
\[C_m^{\perp}=C_{n+2g-2-m}\quad {\rm
and}\quad C_m^q\subseteq C_{mq}.\]

\end{proof}

\subsection{AG codes from the second maximal curve}
By abuse of notations, we still use the same notations as in the previous section for our second maximal curve and corresponding AG codes.

Let $q$ be an odd power of $2$. Thus, $3$ divides $q+1$. Let $F=\F_{q^2}(\cX)$ be the function field of $\cX$ over $\F_{q^2}$, where $\cX$ is defined by the  following equation
\[y^q+y=x^{3}.\]
The genus $g$ of $\cX$ is $g=q-1$ and the number of rational points is $3q^2-2q+1$. The set of these $3q^2-2q+1$  rational points consists of a point at infinity $P_\infty$ and the other  $3q^2-2q$ ``finite" rational points.

Let $n=3q^2-2q$ and let $\{P_1,\dots, P_n\}$ be all $n$ ``finite" rational points. Put
$D=P_1+\dots+P_n$.

\begin{lem}\label{3.4} For a positive integer $m$,
the Euclidean dual $C_\cL(D,m\Pin)^{\perp}$ of $C_\cL(D,m\Pin)$ is $C_\cL(D,(n+2g-2-m)\Pin)$ .
\end{lem}
\begin{proof}
Let $\Ga$ be a $(q^2-1)$th primitive root of unity in  $\F_{q^2}$ and define the polynomial
{\small \[h(x):=x\prod_{j=0}^{3(q-1)-1}\left(\Ga^{j(q+1)/3}-x\right)=x\left(1-x^{3(q-1)}\right)=x-x^{3q-2}.\]}
It is easy to see that $x-\Ga^i$ splits completely in $F$ if and only if $i$ is divisible by $(q+1)/3$. Furthermore, $x$ splits completely in $F$. This implies that the principal divisor ${\rm div}\left(h(x)\right)$
is $D-(3q^2-2q)\Pin$.

Consider the differential $\eta=\frac{dx}{h(x)}$. Then one can verify that ${\rm div}(\eta)=-D+(n+2g-2)\Pin$ and $\res_{P_i}(\eta)=1$ for all $i=1,\dots,n$. Thus, by Proposition \ref{2.3}, we have
\begin{eqnarray*}
C_\cL(D,m\Pin)^{\perp}&=&C_\GO(D,m\Pin)\\
&=&C_\cL(D,D-m\Pin+{\rm div}(\eta))\\
&=&C_\cL(D,(n+2g-2-m)\Pin).\end{eqnarray*}
This completes the proof.
\end{proof}

For simplicity, let us denote by $C_m$ the AG code $C_\cL(D,m\Pin)$. Then, the above result says that $C_m^{\perp}=C_{n+2g-2-m}$. Hence, Lemma \ref{3.4} gives the following results.
\begin{cor}\label{3.5} $C_m$ is Euclidean self-orthogonal if $m\le n/2+g-1$.
\end{cor}

\begin{thm}\label{3.6}
$C_m$ is Hermitian self-orthogonal if  $m\leq 3q-4$.
\end{thm}
\begin{proof} If $m\le 2q-2$, then we have $mq\le n+2g-2-m$. Thus, one has $C_{mq}\subseteq C_{n+2g-2-m}$. Hence, the desired result follows from the fact that
\[C_m^{\perp}=C_{n+2g-2-m}\quad {\rm
and}\quad C_m^q\subseteq C_{mq}.\]
This completes the proof.
\end{proof}

\section{Quantum stabilizer codes}
In this section, we apply the Hermitian self-orthogonality of the classical AG codes $C_m$ constructed in the previous section to produce quantum stabilizer codes and then analyze their parameters.

Let us first recall a result on quantum codes obtained from  Hermitian self-orthogonal classical codes.
\begin{lem}\label{4.1}(see \cite{Ash Kni}) There is a $q$-ary $[[n,n-2k, d^{\perp}]]$-quantum stabilizer code whenever there exists a
$q$-ary classical Hermitian self-orthogonal  $[n,k]$-linear code with dual distance $d^{\perp}$.
\end{lem}

Using the connection of quantum codes with classical Hermitian self-orthogonal codes in Lemma \ref{4.1}, we can derive our main result stated as below. Then we use some numerical results to show that the quantum codes produced from our results are indeed good.
\begin{exm}
\begin{thm}\label{4.2}
If $q$ is a power of $2$, then there exists a $q$-ary $[[2q^2,k_Q:=2q^2-2m+q-2,  d_Q\ge m+2-q]]_q$ quantum code  for any positive integer $m$ satisfying $q-1\le m \leq 2q-2$.
\end{thm}

\begin{thm}\label{4.3}
If $q$ is an odd power of $2$, then there exists a $q$-ary $[[3q^2-2q^2,k_Q:=3q^2-2m-4,  d_Q\ge m+4-2q]]_q$ quantum code  for any positive integer $m$ satisfying $2q-3\le m \leq 3q-4$.
\end{thm}
The proof of Theorems \ref{4.2} and \ref{4.3} directly follows from Theorems \ref{3.3}, \ref{3.6} and Lemma \ref{4.1}.

For $q=2$ and $1\leq m\leq 2$, by Theorem \ref{4.2} we can obtain binary quantum codes with parameters $[[8,4,2]]_2$ and $[[8,2,3]]_2$ which are optimal from the online table \cite{Gr12}.
\end{exm}

\begin{exm}
For $q=4$ and $3\leq m\leq 6$, Theorem \ref{4.2} produces $4$-ary $[[32,34-2m,m-2]]_4$ quantum codes.  Namely, $[[32,28,1]]_4$, $[[32,26,2]]_4$, $[[32,24,3]]_4$, $[[32,22,4]]_4$ quantum codes  can be derived. These codes have good parameters. For instance, in the online table \cite{Br12}, a $[[36,22,4]]_4$ quantum code is given. This implies that our quantum code has a smaller length  for the same dimension and distance.
\end{exm}

\begin{exm}
Let $q=8$ and $7\leq m\leq 14$. Then by Theorem \ref{4.2}, we can derive  $8$-ary $[[126, 134-2m, m-6]]_8$ quantum codes. For instance, new quantum codes with parameters $[[128,108,6]]_8$, $[[128,106,7]]_8$, $[[128,104,8]]_8$ can be produced. They have reasonably better parameters compared with the  quantum codes with parameters $[[134,108,6]]_8$, $[[134,106,7]]_8$, $[[134,96,8]]_8$ given in \cite{Br12}.
\end{exm}

\begin{exm}
Let $q=8$ and $13\leq m\leq 20$. Then we can derive  $8$-ary $[[176, 188-2m, m-12]]_8$ quantum codes. For instance, new quantum codes with parameters $[[176,154,5]]_8$, $[[176,152,6]]_8$, $[[176,150,7]]_8$, $[[176,148,8]]_8$ can be produced. They have reasonably better parameters compared with the  quantum codes with parameters $[[185,149,5]]_8$, $[[185,125,7]]_8$, $[[185,113,8]]_8$ given in \cite{Br12}.
\end{exm}

The above examples show that we can derive quantum codes form Theorem \ref{4.2} which are optimal or even have better parameters compared with \cite{Br12, Gr12}. However, for large $q$, it is difficult to find explicit known codes to compare with ours since there are no suitable tables for reference. Nevertheless, we can still illustrate our result by comparing it with some bounds for large $q$. We only discuss the quantum codes given in Theorem \ref{4.2}.

\begin{rem}\label{4.3} Let us analyze the parameters of the quantum codes given in Theorem \ref{4.2}.
\begin{itemize}
\item[(i)] From the quantum Singleton bound and Theorem \ref{4.2}, the quantum codes given in Theorem \ref{4.2} satisfy
\[n+2-q\le k_Q+2d_Q\le n+2,\]
where $n$ is the length $2q^2$. So the difference of our quantum codes from the Singleton bound is $q$.
\item[(ii)] Let us  consider the quantum Hamming bound \cite{Ket Kla}
\[q^{n-k_Q}\ge \sum_{j=0}^{\lfloor(d_Q-1)/2\rfloor}{n\choose j}(q^2-1)^j.\]
For instance, we just consider the case where $m=2q-3$. Then, $d_Q=q-1$. Thus,  if take logarithm of the right-hand side of the above Hamming bound, we get  the following limit
\[\frac 1q\log_q\left(\sum_{j=0}^{(q-2)/2}{n\choose j}(q^2-1)^j\right)\rightarrow \frac 32\]
as $q$ tends to $\infty$, i.e., the right-hand side of the above Hamming bound is $q^{3q/2+o(q)}$.
The left-hand side of the above Hamming bound is $q^{2q-2}$. If we take logarithm of both the sides with  base $q$, then one can see the difference is about $q/2+o(q)$.  This difference is smaller than the one compared with the Singleton bound.
\end{itemize}
\end{rem}

Lingfei JIN received her Ph.D degree in mathematics from Nanyang Technological University, Singapore in 2013. She is currently a research fellow in Nanyang Technological University, Singapore. Her research interests include classical and quantum coding.

\end{document}